\documentclass[a4paper,reqno]{amsart}
\usepackage{amsmath,amssymb,mathrsfs}
\usepackage{pgf}
\usepackage{tikz}

\usepackage{txfonts,bm,pifont}
\usepackage{cite}
\usepackage{float}
\usepackage{stackrel}
\usepackage{subcaption}
\usepackage{graphicx,xcolor}
\usepackage{comment}
\usepackage{enumitem}
\usepackage{longtable}
\usepackage{hyperref}

\newtheorem{theorem}{Theorem}[section]

\newtheorem{lemma}[theorem]{Lemma}
\newtheorem{remark}[theorem]{Remark}
\theoremstyle{definition}
\newtheorem{definition}[theorem]{Definition}

\numberwithin{equation}{section}
\numberwithin{figure}{section}
\numberwithin{table}{section}
\newcommand{\wutilde}[1]{\vrule depth 0pt width 0pt%
{\raise0.8pt\hbox{$\smash{{\mathop{#1} \limits_{\displaystyle\widetilde{}}}}$}}}
\newcommand{\wuhat}[1]{\vrule depth 0pt width 0pt%
{\raise0.6pt\hbox{$\smash{{\mathop{#1} \limits_{\displaystyle\widehat{}}}}$}}}

\newcommand{\al}{\alpha}
\newcommand{\be}{\beta}
\newcommand{\de}{\delta}
\newcommand{\ga}{\gamma}

\newcommand{\la}{\lambda}

\newcommand{\PDE}{P$\Delta$E}
\newcommand{\bbZ}{\mathbb{Z}}

\newcommand{\bbC}{\mathbb{C}}

\makeatletter
\long\def\@makecaption#1#2{
 \vskip 10pt
 \setbox\@tempboxa\hbox{#1. #2}
 \ifdim \wd\@tempboxa >\hsize #1. #2\par \else \hbox
to\hsize{\hfil\box\@tempboxa\hfil}
 \fi}
\makeatother
\newcommand{\orcidauthorA}{0000-0001-7504-4444}

\begin{document}
\allowdisplaybreaks

\title[3D Consistency of Hirota's dKdV equation]{On the three-dimensional consistency of Hirota's discrete Korteweg-de Vries Equation}
\author{Nalini Joshi}
\address{School of Mathematics and Statistics F07, The University of Sydney, NSW 2006, Australia.}
\thanks{Corresponding author: Nalini Joshi. NJ's ORCID ID is \orcidauthorA. Her research was supported by an
  Australian Research Council Discovery Projects \# DP160101728 and \#DP200100210.}
\email{nalini.joshi@sydney.edu.au}
\author{Nobutaka Nakazono}
\address{Institute of Engineering, Tokyo University of Agriculture and Technology, 2-24-16 Nakacho Koganei, Tokyo 184-8588, Japan.}
\thanks{NN's research was supported by a JSPS KAKENHI Grant Number JP19K14559.}
\email{nakazono@go.tuat.ac.jp}
\begin{abstract}
Hirota's discrete Korteweg-de Vries equation (dKdV) is an integrable partial difference equation on $\bbZ^2$, which approaches the Korteweg-de Vries equation in a continuum limit. 
We find new transformations to other equations, including a second-degree second-order partial difference equation, which provide an unusual embedding into a three-dimensional lattice. 
The consistency of the resulting system extends a property that has been widely used to study partial difference equations on multidimensional lattices.  
\begin{center}
\textit{Dedicated to Harvey Segur on the occasion of his 80th birthday.}
\end{center}
\end{abstract}


\maketitle
\section{Introduction}\label{Introduction}
The main subject of this paper is the partial difference equation
\begin{equation}\label{eqn:dkdv_u_lm}
 u_{l+1,m+1}-u_{l,m}=\dfrac{~1~}{u_{l,m+1}}-\dfrac{~1~}{u_{l+1,m}},\quad
 u\in \bbC,~
 (l, m)\in\bbZ^2.
\end{equation}
This equation is interesting from at least two points of view: (i) it is an integrable discrete version of the Korteweg-de Vries equation and shares many of its famous properties; but, (ii) it does not satisfy a consistency property that has been widely used to characterise and classify many other integrable partial difference equations. 
We describe new observations that overcome this disjunction.

The nonlinear equation \eqref{eqn:dkdv_u_lm} is said to be integrable because it arises as a compatibility condition for a pair of linear partial difference equations, called a Lax pair, depending non-trivially on a spectral parameter. 
On the discrete lattice $\bbZ^2$, compatibility is equivalent to consistency, i.e., the result is independent of the order in which we compose the iterations $l\mapsto l+1$, $m\mapsto m+1$. (For the meaning of terms such as ``iteration" in relation to difference equations, see \cite[\S 1.3]{hietarinta2016discrete}.)
There exist transformations of integrable equations that give rise to iterations in additional directions independent of $(l, m)\in\bbZ^2$.
A natural question is to ask whether the resulting system of equations is consistent in such an extended multi-dimensional lattice.

We answer this question by providing new transformations of Equation \eqref{eqn:dkdv_u_lm} and showing that the resulting system of nonlinear partial difference equations (or, \PDE s) is consistent in a way that appears not to have been observed previously.
It is notable that the system composed of Equation \eqref{eqn:dkdv_u_lm} and its image \PDE s (the outputs arising from transformations) contains a second-degree {\PDE} (see Equation \eqref{eqn:intro_multiquadratic_1}).
We note that a system of second-degree {\PDE}s with two parameters, which contains a case that reduces to Equation \eqref{eqn:dkdv_u_lm}, was studied in \cite{atkinson2012multidimensionally}, from a different point of view concerning factorizability of discriminants.

There is a non-autonomous generalization of Equation \eqref{eqn:dkdv_u_lm} with similar properties. 
It is given by \cite{CNP1991:MR1111648,TGR2001:REMBLAY2001319,kajiwara2008bilinearization}
\begin{equation}\label{eqn:deauto_dkdv_u}
 u_{l+1,m+1}-u_{l,m}=\dfrac{q_{m+1}-p_l}{u_{l,m+1}}-\dfrac{q_m-p_{l+1}}{u_{l+1,m}},\quad
 (p_{l}, q_m)\in\bbC^2,
\end{equation}
where $p_l$, $q_m$ are parameters that vary with one independent variable, either $l$ or $m$ respectively. 
Equation \eqref{eqn:dkdv_u_lm} is given by the reduction $q_{m+1}-p_l=1$, $q_m-p_{l+1}=1$, 
which are satisfied by $p_l=a_0$, $q_m=a_0+1$, for an arbitrary constant $a_0$.

Equation \eqref{eqn:dkdv_u_lm} was proposed by Hirota \cite{HirotaR1977:MR0460934} as a difference analogue of the Korteweg-de Vries (KdV) equation
\begin{equation}\label{eqn:kdv}
 w_t+6\,w\,w_x+w_{xxx}=0,\quad 
 w\in\bbC,\,
 (x, t)\in\bbC^2,
\end{equation}
in the sense that the continuum limit of Equation \eqref{eqn:dkdv_u_lm} leads to the KdV equation (see \cite[\S 2]{HirotaR1977:MR0460934}).
Since Equation \eqref{eqn:dkdv_u_lm} is a special case of Equation \eqref{eqn:deauto_dkdv_u}, and they share similar properties, we refer interchangebly to one or the other as the discrete KdV (dKdV) equation, where there is no confusion. 

Equation \eqref{eqn:dkdv_u_lm} shares distinctive properties with the KdV equation. 
It possesses $N$-soliton solutions, where $N\in \mathbb N\backslash \{0\}$. 
There exist transformations of the dependent variable $u_{l,m}$ that are discrete analogues of B\"acklund transformations of the KdV equation. 
The new transformations we provide are new B\"acklund transformations. 
Equation \eqref{eqn:dkdv_u_lm} has a Lax pair, which gives rise to a discrete inverse scattering transform methodology. 
Our results provide a new way of deducing a Lax pair. 

There is a graphical interpretation of consistency that provides a simple geometric way of studying the evolution of the dKdV equation (see \S1.3 of \cite{hietarinta2016discrete}). 
Consider each iterate $\{u_{l,m}$, $u_{l+1,m}$, $u_{l,m+1}$, $u_{l+1,m+1}\}$ as the value assigned to a vertex of a unit cell in a lattice isomorphic to $\bbZ^2$; see Figure \ref{fig:square}. 
(In the non-autonomous case, we assign the parameter $p_l$ to the edge connecting $u_{l,m}$ with $u_{l+1,m}$ and $q_m$ to the edge connecting $u_{l,m}$ with $u_{l,m+1}$.)
\begin{figure}[htbp]
 \begin{center}
 \includegraphics[width=0.4\textwidth]{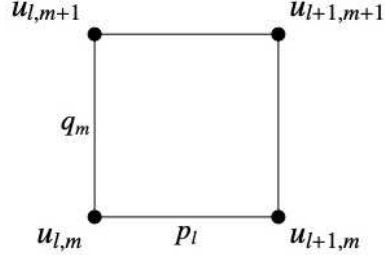}
 \end{center}
\caption{The iterates $u_{l,m}$ and parameters $(p_l, q_m)$ assigned to a unit cell of $\bbZ^2$.}
\label{fig:square}
\end{figure}
Then, the iteration $l\mapsto l+1$ is equivalent to mapping the initial cell to an adjacent cell in the horizontal direction in Figure \ref{fig:square}, while iterating $m\mapsto m+1$ is equivalent to moving in the vertical direction.

We now extend this two-dimensional lattice to a third dimension in the following way (see \S 3.2 of \cite{hietarinta2016discrete} for mathematical details).
Given a unit cell as shown in Figure \ref{fig:square}, take a unit cell in an independent lattice and identify two of its vertices, and the edge joining these two, with those in the original cell. 
We interpret the remaining edges in the new cell as lying along a third direction, independent of the two generating the original lattice. 
For the non-autonomous equation, we also require a new parameter that evolves in the new direction.

This generates a three-dimensional lattice given by $\bbZ^3$, with unit cells given by cubes. 
If we are given a value at a vertex lying in the third direction, then the corresponding three-dimensional initial value problem for the \PDE\ starts with four initial values. 
If we assume that the \PDE\ is given on each face of the cube, then we can find the values at every vertex of the resulting cube by solving the \PDE.
 
While this construction is straightforward, there are possible inconsistencies that may arise because a vertex value on the cube may be calculated in three different ways, corresponding to the three different faces meeting at that vertex. 
Lattice equations that satisfy the resulting consistency conditions are said to be \textit{consistent around a cube} (CAC). 
Equations satisfying the CAC property have been studied and classified under certain conditions \cite{NW2001:MR1869690,ABS2003:MR1962121}. 
But, as mentioned above, the dKdV equation is known to not satisfy these consistency conditions \cite[\S 3.8.1]{hietarinta2016discrete}. 

The main new result of the current paper is to demonstrate a different construction by which the dKdV equation is consistently embedded in a higher dimensional lattice $\bbZ^3$. 
Although the lattice structure we discover is different to the one conventionally used in the literature, we show that it is not unique to the dKdV equation $\--$ the same type of structure arises in at least one other integrable \PDE, known as the lattice sine-Gordon equation\cite{volkov1992quantum,bobenko1993discrete}:
\begin{equation}\label{eqn:redbook390_u}
 \dfrac{~u_{l+1,m+1}~}{u_{l,m}}
 =\dfrac{(\ga u_{l,m+1}-1)(\ga-u_{l+1,m})}{(\ga-u_{l,m+1})(\ga u_{l+1,m}-1)},\quad
 \ga\in\bbC.
\end{equation}

The list of equations which have the CAC property is given in \cite{ABS2003:MR1962121,ABS2009:MR2503862,BollR2011:MR2846098,BollR2012:MR3010833}.
As remarked earlier, the dKdV equation \eqref{eqn:dkdv_u_lm} is not in this list. Similarly,  
the non-autonomous dKdV equation \eqref{eqn:deauto_dkdv_u}
and 
the lattice sine-Gordon equation \ref{eqn:redbook390_u}
do not appear in this list.
The relation between Equations \eqref{eqn:dkdv_u_lm} and \eqref{eqn:redbook390_u} with equations in the list can be found in \cite[\S\S 3.8.1, 3.8.3]{hietarinta2016discrete},
while the relation for Equation \eqref{eqn:deauto_dkdv_u} is given in \cite{CNP1991:MR1111648}.
\subsection{Multi-quadratic equations} Our results involve new transformations to multi-\-quad\-ratic equations, involving a parameter $\lambda\in\mathbb C$, which are listed below.
The first group of transformations given by Equations \eqref{eqns:dkdv_uv_P234} maps the autonomous dKdV equation \eqref{eqn:dkdv_u_lm} to
\begin{equation}\label{eqn:intro_multiquadratic_1}
\begin{split}
 &\left(v_{l,m}\,v_{l+1,m}-v_{l,m+1}\,v_{l+1,m+1}\right)^2\\
 &\qquad +\bigl(v_{l,m}-v_{l+1,m+1}\bigr)\bigl(v_{l+1,m}-v_{l,m+1}\bigr)\bigl(\lambda-v_{l,m}\,v_{l+1,m}\bigr)\bigl(\lambda-v_{l,m+1}\,v_{l+1,m+1}\bigr)=0.
\end{split}
\end{equation}
The second group of transformations (see Equations \eqref{eq:fulltransformations}) maps the non-autonomous dKdV equation \eqref{eqn:deauto_dkdv_u} to 
\begin{equation}\label{eqn:intro_multiquadratic_2}
\begin{split}
 &q_m\left(v_{l,m}\,v_{l+1,m}-v_{l,m+1}\,v_{l+1,m+1}\right)^2\\
 &\quad +\bigl(v_{l,m}-v_{l+1,m+1}\bigr)\bigl(v_{l+1,m}-v_{l,m+1}\bigr)\bigl(\lambda-v_{l,m}\,v_{l+1,m}\bigr)\bigl(\lambda-v_{l,m+1}\,v_{l+1,m+1}\bigr)\\
 &\quad +\bigl(p_{l+1}v_{l,m}-p_lv_{l+1,m+1}\bigr)\bigl(p_lv_{l+1,m}-p_{l+1}v_{l,m+1}\bigr)\\
 &\quad -(p_l+p_{l+1})\bigl(\lambda\,v_{l,m}\,v_{l+1,m}+\lambda\,v_{l,m+1}\,v_{l+1,m+1}-2\,v_{l,m}\,v_{l+1,m}\,v_{l,m+1}\,v_{l+1,m+1}\bigr)\\
 &\quad+\bigl(p_{l+1}\,v_{l,m}\,v_{l,m+1}+p_l\,v_{l+1,m}\,v_{l+1,m+1}\bigr)\bigl(2\lambda - v_{l,m} v_{l+1,m}-v_{l,m+1} \,v_{l+1,m+1}\bigr) =0.
\end{split}
\end{equation}
The third group of transformations (see Equations \eqref{eq:redbooktransf}) maps the lattice sine-Gordon equation \eqref{eqn:redbook390_u} to 
\begin{equation}\label{eqn:intro_multiquadratic_3}
\begin{split}
 &\gamma^2\left(v_{l,m}-\,v_{l+1,m+1}\right)\bigl(v_{l+1,m}-v_{l,m+1}\bigr)\bigl(\lambda+v_{l,m}v_{l+1,m}\bigr)\bigl(\lambda+v_{l,m+1}v_{l+1,m+1}\bigr)\\
 &\quad -\,(1-\gamma^2)(\gamma^2-\,\lambda)\bigl(v_{l,m}v_{l+1,m}-\,v_{l,m+1}v_{l+1,m+1}\bigr)^2\\
 &\quad +\gamma(1-\gamma^2)\bigl(v_{l,m}v_{l+1,m}-\,v_{l,m+1}v_{l+1,m+1}\bigr)\\
 &\qquad \times\, \Bigl(\bigl(\lambda+v_{l,m}v_{l+1,m}\bigr)\bigl( v_{l,m+1}+v_{l+1,m+1}\bigr) -\,\bigl(v_{l,m}+v_{l+1,m}\bigr)\bigl(\lambda+v_{l,m+1}v_{l+1,m+1}\bigr)\Bigr)=0.
\end{split}
\end{equation}
\subsection{Notation and Terminology}\label{subsection:notation_definitions}
For conciseness in the remainder of the paper, we adopt the following notation:
\begin{equation}\label{eq:pde terminology_1}
 u=u_{l,m},\quad
 \overline u=u_{l+1,m},\quad
 \widetilde u=u_{l,m+1},\quad
 \underline u=u_{l-1,m},\quad 
 \wutilde u=u_{l,m-1},\quad
 \widetilde{\overline{u}}=u_{l+1,m+1},
\end{equation}
and extend the notation to $p_l$, $q_m$ and other iterates of $u$ as needed. 

We write each lattice equation as the vanishing condition of a polynomial of four variables.
For example, the dKdV equation is given by $Q(u,\overline{u},\widetilde{u},\widetilde{\overline{u}})=0$, where
\[
 Q(u,\overline{u},\widetilde{u},\widetilde{\overline{u}})
 =\overline{u}\,\widetilde{u}(\widetilde{\overline{u}}-u)-(q_{m+1}-p_l)\overline{u}+(q_m-p_{l+1})\widetilde{u}.
\]
(Where convenient, we also use lattice equations in their equivalent rational forms.)
Note that, for conciseness, we omit the dependence of the polynomial $Q$ on parameters. 
We assume that any parameters in the polynomial take generic values and that the corresponding polynomial is irreducible. 
Where the corresponding polynomial is linear in each variable, we describe it as an affine linear polynomial. 
Where the polynomial is quadratic in each variable, we refer to it as a multi-quadratic polynomial. 

Because of the association with a quadrilateral of $\bbZ^2$, see Figure \ref{fig:square}, a lattice equation relating four vertex values is called a \textit{quad-equation}. 
By a small abuse of terminology, we will also refer to the corresponding function, whose vanishing condition gives the lattice equation, as a quad-equation.

\subsection{Outline of the paper}
In Section \ref{s:3D}, we describe a new way of embedding a lattice equation in $\bbZ^3$, which differs from the conventional one used for the CAC property. 
This process is applied to the autonomous dKdV equation \eqref{eqn:dkdv_u_lm} in Section \ref{section:props_dkdv}, where we also show how to deduce a Lax pair for this equation and how it is related to Equation \eqref{eqn:intro_multiquadratic_1}. 
Finally, in Section \ref{s:fulldkdv}, we show how to extend the construction to the non-autonomous dKdV equation \eqref{eqn:deauto_dkdv_u} and to the lattice sine-Gordon equation \eqref{eqn:redbook390_u}.

\section{Embedding into three-dimensions}\label{s:3D}
In this section, we describe a way to embed the dKdV equation in $\bbZ^3$. 
To each elementary cubic cell in $\bbZ^3$, we associate 8 variables denoted by
\[
(u_0,u_1,u_2,u_{12},v_0,v_1,v_2,v_{12})\in\bbC^8,
\]
and assign each variable to a vertex of the cube.

In contrast to the usual procedure assumed for proving the CAC property, we do not assign a quad-equation to each face of the cube. 
Instead, we describe a system of equations on the cube, which may (i) vary with each face; (ii) become a triangular equation, i.e., those relating only three vertex values, on certain faces; and, (iii) involves vertices of a quadrilateral given by an interior diagonal slice of the cube. 

Three of the quad-equations occur on the bottom, front and back faces of the cube, while the fourth one occurs in the interior of the cube as a diagonal slice. 
Each triangular domain occurs as a half of the left or right face of the cube. 
See Figure \ref{fig:cubeABCS}. 
We will refer to this configuration as a \textit{broken cube}.
\begin{figure}[H]
\begin{center}
 \includegraphics[width=0.4\textwidth]{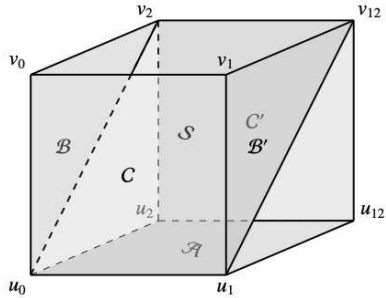}
\end{center}
\caption{A cube with three quadrilateral faces labelled by ${\mathcal A}$, ${\mathcal C}$ and ${\mathcal C}'$, an interior diagonal quadrilateral labelled by ${\mathcal S}$ and triangular domains labelled as ${\mathcal B}$ and ${\mathcal B}'$. Note that primes denote domains on parallel faces.}
\label{fig:cubeABCS}
\end{figure}
Correspondingly, we define polynomials of 4 variables $\mathcal A$, $\mathcal S$, $\mathcal C$, $\mathcal C'$ $: \bbC^4\to \bbC$ and those of 3 variables $\mathcal B$, $\mathcal B'$ $:\bbC^3\to \bbC$, such that $\mathcal B$ and $\mathcal B'$ written as functions of $(x, y, z)$ satisfy
\begin{enumerate}[leftmargin=0.8cm,label={\rm(\roman*)}]
\item 
$\deg_x{\mathcal B}\geq1$,~ 
$\deg_y{\mathcal B}=\deg_z{\mathcal B}=1$;
\item
the equation ${\mathcal B}=0$ can be solved for $y$ and $z$, and each solution is a rational function of the other two arguments.
\end{enumerate}
With the labelling of vertices given in Figure \ref{fig:cubeABCS}, we denote the system of six corresponding equations by
\begin{subequations}\label{eq:face-eqns_localABCS}
\begin{align}
 &{\mathcal A}(u_0,u_1,u_2,u_{12})=0,
 &&{\mathcal S}(u_0,u_1,v_2,v_{12})=0,\\
 &{\mathcal B}(u_0,v_0,v_2)=0,
 &&{\mathcal B}'(u_1,v_1,v_{12})=0,\\
 &{\mathcal C}(u_0,u_1,v_0,v_1)=0,
 &&{\mathcal C}'(u_2,u_{12},v_2,v_{12})=0.
\end{align}
\end{subequations}
The following definition describes how consistency holds for this system of equations.

\begin{definition}[CABC property]\label{def:CABC_local}\rm
Let $\{u_0,u_1,u_2,v_0\}$ be given initial values. 
Using Equations \eqref{eq:face-eqns_localABCS},
we can express the variable $v_{12}$ as a rational function in terms of the initial values in 3 ways.
When the 3 results for $v_{12}$ are equal, the system of Equations \eqref{eq:face-eqns_localABCS} is said to be \emph{consistent around a broken cube} or to have the \emph{consistency around a broken cube} (CABC) property.
In this case, we refer to the configuration of quadrilaterals and triangular domains associated with the polynomials $\mathcal A$, $\mathcal S$, $\mathcal C$, $\mathcal C'$, $\mathcal B$, $\mathcal B'$ as a {\it CABC cube}.
\end{definition}

\begin{remark}
Here is a brief explanation of how $v_{12}$ is calculated.
Using ${\mathcal A}=0$, ${\mathcal B}=0$ and ${\mathcal C}=0$ we can respectively express $u_{12}$, $v_2$ and $v_1$ as rational functions in terms of the initial values.
Then, using the remaining three equations we can express $v_{12}$ as a rational function in terms of the initial values in 3 ways.
Alternatively, $v_{12}$ can also be expressed as follows.
Using ${\mathcal A}=0$ and ${\mathcal C}'=0$ we can obtain ${\mathcal S}=0$.
Then, we now essentially have five relations, not six.
The method to determine $u_{12}$, $v_2$ and $v_1$ is the same as above. 
However, there are two ways to determine $v_{12}$ afterwards, using ${\mathcal B}'=0$ and ${\mathcal C}'=0$.
\end{remark}

Other equations arise from interrelationships between the above equations on the broken cube. 
For example, we show in the next section that Equation \eqref{eqn:intro_multiquadratic_1} arises on the top face, parallel to $\mathcal A$. 
It is also useful to note equations that relate three vertices on a face to a vertex on the opposite face. 
The following definition of such equations uses terminology analogous to existing ones in the literature on the CAC property. 

\begin{definition}[Tetrahedron property]\label{def:tetrahedron_local}\rm
A CABC cube is said to have a \emph{tetrahedron property},
if there exist quad-equations ${\mathcal K}_1$ and ${\mathcal K}_2$ satisfying
\begin{equation}
 {\mathcal K}_1(u_0,u_1,v_0,v_{12})=0,\quad
 {\mathcal K}_2(u_0,u_1,v_1,v_2)=0.
\end{equation}
In this case, each of the equations ${\mathcal K}_1=0$ and ${\mathcal K}_2=0$ is referred to as a \emph{tetrahedron equation}.
\end{definition}

\begin{remark}
The above description of the broken cube (or the CABC property) and its iteration in three-dimensional space cannot be replaced by affine linear transformations of the standard cubic lattice. 
It is possible to reflect a broken cube around a horizontal plane through its centre and place a copy of it above the original one to create a vertical column of alternating broken cubes. 
Stacking adjacent such vertical columns side-by-side produces a structure like that shown in Figure \ref{fig:stacks}. 
However, it is not possible to rotate each alternating column around its vertical axis of symmetry to create diagonal unit cubes from the triangular equations, because the iteration of the above polynomials does not have the required symmetry.
\end{remark}

\begin{figure}[H]
\begin{center}
 \includegraphics[width=0.4\textwidth]{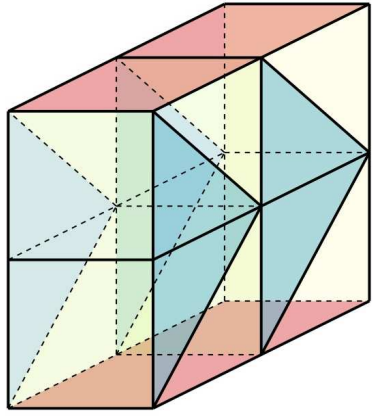}
\end{center}
\caption{A stacking of vertically alternating broken cubes.}
\label{fig:stacks}
\end{figure}

By interpreting each vertex value as an iterate of a function in an appropriate way, we can interpret the above equations as \PDE s. 
In particular, we use the terminology given in Equation \eqref{eq:pde terminology_1} for $u_{l,m}$ and extend it to $v_{l,m}$ to give the following definition of \PDE s.

\begin{definition}[CABC and tetrahedron properties for a system of {\PDE}s]
Define the \PDE s
\begin{equation}\label{eqn:general_P1234}
 {\mathsf A}\bigl(u,\overline{u},\widetilde{u},\widetilde{\overline{u}}\bigl)=0,\quad
 {\mathsf S}\big(u,\overline{u},\widetilde{v},\widetilde{\overline{v}}\big)=0,\quad
 {\mathsf B}\big(u,v,\widetilde{v}\big)=0,\quad
 {\mathsf C}\big(u,\overline{u},v,\overline{v}\big)=0,
\end{equation}
which give the following equations around each elementary cubic cell in $\bbZ^3$:
\begin{subequations}\label{eqn:general_ASBBCC}
\begin{align}
 &\mathcal A={\mathsf A}\bigl(u,\overline{u},\widetilde{u},\widetilde{\overline{u}}\bigl)=0,
 &&\mathcal S={\mathsf S}\big(u,\overline{u},\widetilde{v},\widetilde{\overline{v}}\big)=0,\\
 &{\mathcal B}={\mathsf B}\big(u,v,\widetilde{v}\big)=0,
 &&{\mathcal B}'={\mathsf B}\big(\overline{u},\overline{v},\widetilde{\overline{v}}\big)=0,\\
 &{\mathcal C}={\mathsf C}\big(u,\overline{u},v,\overline{v}\big)=0,
 &&{\mathcal C}'={\mathsf C}\big(\widetilde{u},\widetilde{\overline{u}},\widetilde{v},\widetilde{\overline{v}}\big)=0.
\end{align}
\end{subequations}
Here, we have used the terminology given in Equation \eqref{eq:pde terminology_1}.
Then, the system \eqref{eqn:general_P1234} is said to have the CABC property if Definition \ref{def:CABC_local} holds for the equations \eqref{eqn:general_ASBBCC}.
We also transfer the definition of tetrahedron properties to \PDE s corresponding to $\mathcal K_j$, $j=1,2$, in the obvious way. 
Moreover, the {\PDE} ${\mathsf A}\bigl(u,\overline{u},\widetilde{u},\widetilde{\overline{u}}\bigl)=0$ will be described as having the CABC property, if the system \eqref{eqn:general_P1234} has the CABC property.
\end{definition}

\begin{remark}
Note that Equations \eqref{eqn:general_P1234} are not necessarily autonomous. 
They may contain parameters that evolve with $(l, m)$.
\end{remark}

\section{Lattice structure of Equation \eqref{eqn:dkdv_u_lm}}
\label{section:props_dkdv}
In this section, we show that Equation \eqref{eqn:dkdv_u_lm} has the CABC property, by deducing the corresponding equations on a broken cube in a three-dimensional lattice.
The main advantage of this embedding is that it enables us to construct its Lax pair. 
We give this construction below. 
We also show that the multi-quadratic quad-equation \eqref{eqn:intro_multiquadratic_1} arises from the CABC property. 

We start by defining equations on the broken cube:
\begin{subequations}\label{eqns:dkdv_uv_P234}
\begin{align}
 &\mathsf A\bigl(u,\overline{u},\widetilde{u},\widetilde{\overline{u}}\bigl)
 =\widetilde{\overline{u}}-u-\dfrac{~1~}{\widetilde{u}}+\dfrac{~1~}{\overline{u}}=0,\label{eqn:dkdv_uv_P1}\\
 &\mathsf S\big(u,\overline{u},\widetilde{v},\widetilde{\overline{v}}\big)
 =u-\widetilde{\overline{v}}-\dfrac{~1~}{\overline{u}}+\dfrac{~\la~}{\widetilde{v}}=0,\label{eqn:dkdv_uv_P2}\\
 &\mathsf B\big(u,v,\widetilde{v}\big)
 =(u-v)\Big(\dfrac{~1~}{u}+\widetilde{v}\Big)-1+\la=0,\label{eqn:dkdv_uv_P3}\\
 &\mathsf C\big(u,\overline{u},v,\overline{v}\big)
 =\dfrac{~1~}{u}-\dfrac{~\la~}{v}-\overline{u}+\overline{v}=0,\label{eqn:dkdv_uv_P4}
\end{align}
\end{subequations}
where, as before, we have used the terminology given in Equation \eqref{eq:pde terminology_1}. 

It is straightforward to confirm that Definition \ref{def:CABC_local} holds when $\mathcal A$, $\mathcal B$, $\mathcal B'$, $\mathcal S$, $\mathcal C$ and $\mathcal C'$ are defined by \eqref{eqn:general_ASBBCC}.
Moreover, the tetrahedron properties hold. 
More explicitly, we have 
\begin{subequations}\label{eqn:eqnsASBC_dkdv}
\begin{align}
 &{\mathcal A}
 =\widetilde{\overline{u}}-u-\dfrac{~1~}{\widetilde{u}}+\dfrac{~1~}{\overline{u}}=0,
 &&{\mathcal S}
 =u-\widetilde{\overline{v}}-\dfrac{~1~}{\overline{u}}+\dfrac{~\la~}{\widetilde{v}}=0,\\
 &{\mathcal B}
 =\Big(u-v\Big)\Big(\dfrac{~1~}{u}+\widetilde{v}\Big)-1+\la=0,
 &&{\mathcal B}'
 =\Big(\overline{u}-\overline{v}\Big)\Big(\dfrac{~1~}{\overline{u}}+\widetilde{\overline{v}}\Big)-1+\la=0,\\
 &{\mathcal C}
 =\dfrac{~1~}{u}-\dfrac{~\la~}{v}-\overline{u}+\overline{v}=0,
 &&{\mathcal C}'
 =\dfrac{~1~}{\widetilde{u}}-\dfrac{~\la~}{\widetilde{v}}-\widetilde{\overline{u}}+\widetilde{\overline{v}}=0,
\end{align}
\end{subequations}
while the tetrahedron equations are given by 
\begin{equation}
 {\mathcal K}_1
 =\Big(\,\widetilde{\overline{v}}+\dfrac{~1~}{\overline{u}}\Big)\Big(\dfrac{~1~}{u}-\dfrac{~\la~}{v}\Big)-1+\la=0,\qquad
 {\mathcal K}_2
 =\Big(\dfrac{~\la~}{\widetilde{v}}+u\Big)\Big(\overline{u}-\overline{v}\Big)-1+\la=0.
\end{equation}
Hence, the following theorem holds.
\begin{theorem}
Equation \eqref{eqn:dkdv_u_lm} has the CABC property.
\end{theorem}

Now we show how to construct a Lax pair for the dKdV equation from the above system of equations, through a method that parallels the well-known method for constructing the Lax pair using the CAC property\cite{BS2002:MR1890049,NijhoffFW2002:MR1912127,WalkerAJ:thesis}.

To carry this out, we represent the auxiliary function $v$ by
\begin{equation}\label{eqn:dkdv_vFG}
 v=\dfrac{F}{G},
\end{equation}
where $F=F_{l,m}$ and $G=G_{l,m}$. 
Substituting this into the equations \eqref{eqn:dkdv_uv_P3} and \eqref{eqn:dkdv_uv_P4},
 separating the numerators and denominators of the resulting equations, and using the 2-vector $\Psi=\Psi_{l,m}$ defined by
\begin{equation}
 \Psi=\begin{pmatrix}F\\G\end{pmatrix},
\end{equation}
we obtain the following linear systems:
\begin{equation}\label{eqn:dkdv_Lax_Psi_delta}
 \overline{\Psi}
 =\de^{(1)}_{l,m}
 \begin{pmatrix}
 \overline{u}-\dfrac{~1~}{u}&\la\\
 1&0
 \end{pmatrix}
 \Psi,\quad
 \widetilde{\Psi}
 =\de^{(2)}_{l,m}
 \begin{pmatrix}
 -\dfrac{~1~}{u}&\la\\
 1&-u
 \end{pmatrix}
 \Psi,
\end{equation}
where $\de^{(1)}_{l,m}$ and $\de^{(2)}_{l,m}$ are arbitrary decoupling factors.
The system of linear equations \eqref{eqn:dkdv_Lax_Psi_delta} is the Lax pair of Equation \eqref{eqn:dkdv_u_lm}.
Indeed, we can easily verify that the compatibility condition $\widetilde{\overline{\Psi}}=\overline{\widetilde{\Psi}}$ gives Equation \eqref{eqn:dkdv_u_lm} and 
\begin{equation}
 \de^{(1)}_{l,m}\de^{(2)}_{l+1,m}=\de^{(1)}_{l,m+1}\de^{(2)}_{l,m}.
\end{equation}
For simplicity, in what follows we take 
\begin{equation}
 \de^{(1)}_{l,m}=\de^{(2)}_{l,m}=1.
\end{equation}

We next show how the multi-quadratic quad-equation \eqref{eqn:intro_multiquadratic_1} arises from the system \eqref{eqns:dkdv_uv_P234}. 
\begin{theorem}\label{theorem:dkdv_mul}
If $u$ and $v$ satisfy the system of equations \eqref{eqns:dkdv_uv_P234}, then $v$ satisfies the multi-quadratic equation \eqref{eqn:intro_multiquadratic_1}. 
\end{theorem}
\begin{proof} 
Eliminating the term $\overline{u}$ from equations \eqref{eqn:dkdv_uv_P2} and \eqref{eqn:dkdv_uv_P4}, we obtain
\begin{equation}
 u-\dfrac{v}{\la-v\overline{v}}+\dfrac{u^2\,\widetilde{v}}{\la-\widetilde{v}\,\widetilde{\overline{v}}}=0.
\end{equation}
Furthermore, eliminating the term $u$ from the equation above and Equation \eqref{eqn:dkdv_uv_P3}, we obtain 
\begin{equation}\label{eqn:dkdv_v_P5}
 {\mathcal A}'\Big(v,\overline{v},\widetilde{v},\widetilde{\overline{v}}\Big)
 =\Big(v\overline{v}-\widetilde{v}\,\widetilde{\overline{v}}\Big)^2+\Big(v-\widetilde{\overline{v}}\Big)\Big(\overline{v}-\widetilde{v}\Big)\Big(\la-v\overline{v}\Big)\Big(\la-\widetilde{v}\,\widetilde{\overline{v}}\Big)=0,
\end{equation}
which is equivalent to Equation \eqref{eqn:intro_multiquadratic_1}.
This completes the proof.
\end{proof}

\begin{remark}
The multivariate polynomial ${\mathcal A}'$ given by Equation \eqref{eqn:dkdv_v_P5} satisfies the Kleinian symmetry:
\begin{equation}\label{eqn:Kleinian symmetry}
 {\mathcal A}'\Big(v,\overline{v},\widetilde{v},\widetilde{\overline{v}}\Big)
 ={\mathcal A}'\Big(\overline{v},v,\widetilde{\overline{v}},\widetilde{v}\Big)
 ={\mathcal A}'\Big(\widetilde{v},\widetilde{\overline{v}},v,\overline{v}\Big),
\end{equation}
and the discriminants of ${\mathcal A}'={\mathcal A}'\Big(v,\overline{v},\widetilde{v},\widetilde{\overline{v}}\Big)$ are given by
\begin{subequations}\label{eqns:A'_discriminants}
\begin{align}
 &\Delta[{\mathcal A}',v]
 =G\Big(\overline{v},\widetilde{v}\Big)H_1\Big(\widetilde{v},\widetilde{\overline{v}}\Big)H_2\Big(\overline{v},\widetilde{\overline{v}}\Big),\\
 &\Delta[{\mathcal A}',\overline{v}]
 =G\Big(v,\widetilde{\overline{v}}\Big)H_1\Big(\widetilde{v},\widetilde{\overline{v}}\Big)H_2\Big(v,\widetilde{v}\Big),\\
 &\Delta[{\mathcal A}',\widetilde{v}]
 =G\Big(v,\widetilde{\overline{v}}\Big)H_1\Big(v,\overline{v}\Big)H_2\Big(\overline{v},\widetilde{\overline{v}}\Big),\\
 &\Delta[{\mathcal A}',\widetilde{\overline{v}}]
 =G\Big(\overline{v},\widetilde{v}\Big)H_1\Big(v,\overline{v}\Big)H_2\Big(v,\widetilde{v}\Big),
\end{align}
\end{subequations}
where the polynomials $G(x,y)$, $H_1(x,y)$ and $H_2(x,y)$ are given by
\begin{equation}
 G(x,y)=(x-y)^2,\quad
 H_1(x,y)=(xy-\la)^2,\quad
 H_2(x,y)=(xy-\la)^2+4xy.
\end{equation}
Note that for a quadratic equation $Q=ax^2+bx+c$ the discriminant $\Delta[Q,x]$ is defined by
\begin{equation}
 \Delta[Q,x]=b^2-4ac.
\end{equation}
From the above information, we find that Equation \eqref{eqn:dkdv_v_P5} is $dQ3^{0\ast}$ type in the sense given in \cite{kassotakis2018difference,atkinson2014multi}.
\end{remark}

We finally show that Equations \eqref{eqn:dkdv_u_lm} and \eqref{eqn:intro_multiquadratic_1} can be solved by the lattice potential mKdV equation \eqref{eqn:dkdv_phi_lmkdv}. 
Rewrite the Lax pair \eqref{eqn:dkdv_Lax_Psi_delta} as follows
\begin{equation}\label{eqn:dkdv_FG_eqns}
 \overline{F}=-\Big(\dfrac{~1~}{u}-\overline{u}\Big)F+\la G,\quad
 \overline{G}=F,\quad
 \widetilde{F}=-\dfrac{~1~}{u}F+\la G,\quad
 \widetilde{G}=F-uG.
\end{equation}
Eliminating the variables $u$ and $F$ from Equations \eqref{eqn:dkdv_FG_eqns} leads to
 the following {\PDE}:
\begin{equation}\label{eqn:dkdv_quad_G}
 \Big(\widetilde{\overline{G}}-\la G\Big)\Big(\overline{G}-\widetilde{G}\Big)+G\overline{G}=0.
\end{equation}
Equation \eqref{eqn:dkdv_quad_G} is known as the lattice potential mKdV equation (lmKdV) \cite{NC1995:MR1329559,NQC1983:MR719638}, also labelled $H3$ in the list of lattice equations with the CAC property \cite{ABS2003:MR1962121}.
Indeed, by setting 
\begin{equation}\label{eqn:dkdv_phiG_eqn}
 \phi_{l,m}=(-\la)^{-l/2}(1-\la)^{-m/2}G_{l,m},
\end{equation}
Equation \eqref{eqn:dkdv_quad_G} can be rewritten in the following canonical form:
\begin{equation}\label{eqn:dkdv_phi_lmkdv}
 \dfrac{~\widetilde{\overline{\phi}}~}{\phi}
 =\dfrac{(1-\la)^{1/2}\overline{\phi}-(-\la)^{1/2}\widetilde{\phi}}{(1-\la)^{1/2}\widetilde{\phi}-(-\la)^{1/2}\overline{\phi}},
\end{equation}
where $\phi=\phi_{l,m}$.
Moreover, from Equations \eqref{eqn:dkdv_vFG}, \eqref{eqn:dkdv_FG_eqns} and \eqref{eqn:dkdv_phiG_eqn}
we obtain 
\begin{subequations}
\begin{align}
 &u=\dfrac{F-\widetilde{G}}{G}
 =\dfrac{\overline{G}-\widetilde{G}}{G}
 =\dfrac{(-\la)^{1/2}\overline{\phi}-(1-\la)^{1/2}\widetilde{\phi}}{\phi},\\
 &v=\dfrac{F}{G}
 =\dfrac{\overline{G}}{G}
 =(-\la)^{1/2}\dfrac{~\overline{\phi}~}{\phi}.
\end{align}
\end{subequations}
Therefore, we obtain the following lemma.
\begin{lemma}
Let $\phi=\phi_{l,m}$ be a solution of Equation \eqref{eqn:dkdv_phi_lmkdv}. 
Then, $u$ and $v$ given by
\begin{equation}
 u=\dfrac{(-\la)^{1/2}\overline{\phi}-(1-\la)^{1/2}\widetilde{\phi}}{\phi},\quad
 v=(-\la)^{1/2}\dfrac{~\overline{\phi}~}{\phi},
\end{equation}
solve Equations  \eqref{eqn:dkdv_u_lm} and \eqref{eqn:intro_multiquadratic_1}, respectively. 
\end{lemma}
\begin{proof}
The statement can be verified by direct calculation.
\end{proof}

\section{Properties of Equations \eqref{eqn:deauto_dkdv_u} and \eqref{eqn:redbook390_u}}\label{s:fulldkdv}
In this section, we show that Equations \eqref{eqn:deauto_dkdv_u} and \eqref{eqn:redbook390_u} have the same properties as Equation \eqref{eqn:dkdv_u_lm}.
The process for demonstrating the result for each equation is exactly the same as that for Equation \eqref{eqn:dkdv_u_lm} discussed in Section \ref{section:props_dkdv} and so, for consciseness, we omit detailed arguments.

\subsection{Properties of Equation \eqref{eqn:deauto_dkdv_u}}
The system of {\PDE}s for Equation \eqref{eqn:deauto_dkdv_u}, which has the CABC and tetrahedron properties, is given by:
\begin{subequations}\label{eq:fulltransformations}
\begin{align}
 &{\mathsf A}=\widetilde{\overline{u}}-u-\dfrac{q_{m+1}-p_l}{\widetilde{u}}+\dfrac{q_m-p_{l+1}}{\overline{u}}=0,\\
 &{\mathsf S}=u-\widetilde{\overline{v}}-\dfrac{q_m-p_{l+1}}{\overline{u}}+\dfrac{\la-p_l}{\widetilde{v}}=0,\\
 &{\mathsf B}=\Big(u-v\Big)\Big(\dfrac{q_m-p_l}{u}+\widetilde{v}\Big)-q_m+\la=0,\\
 &{\mathsf C}=\dfrac{q_m-p_l}{u}-\dfrac{\la-p_l}{v}-\overline{u}+\overline{v}=0.
\end{align}
\end{subequations}
Moreover, the tetrahedron equations are given by 
\begin{subequations}
\begin{align}
 &{\mathcal K}_1
 =\Big(\widetilde{\overline{v}}+\dfrac{q_m-p_{l+1}}{\overline{u}}\Big)\Big(\dfrac{q_m-p_l}{u}-\dfrac{\la-p_l}{v}\Big)-q_m+\la
 =0,\\
 &{\mathcal K}_2
 =\Big(\dfrac{\la-p_l}{\widetilde{v}}+u\Big)\Big(\overline{u}-\overline{v}\Big)-q_m+\la
 =0,
\end{align}
\end{subequations}
and the Lax pair of Equation \eqref{eqn:deauto_dkdv_u} is given by
\begin{equation}\label{eqn:deauto_dkdv_Lax_Psi}
 \overline{\Psi}
 =\begin{pmatrix}
 \overline{u}-\dfrac{q_m-p_l}{u}&\la-p_l\\
 1&0
 \end{pmatrix}
 \Psi,\quad
 \widetilde{\Psi}
 =\begin{pmatrix}
 -\dfrac{q_m-p_l}{u}&\la-p_l\\
 1&-u
 \end{pmatrix}
 \Psi.
\end{equation}

\begin{theorem}
If $u$ and $v$ satisfy the system of equations \eqref{eq:fulltransformations}, then $v$ satisfies the multi-quadratic equation \eqref{eqn:intro_multiquadratic_2}. 
\end{theorem}
\begin{proof}
The method of proof is the same as in the proof of Theorem \ref{theorem:dkdv_mul}.
\end{proof}

As for Equation \eqref{eqn:dkdv_u_lm}, there is a relation between the solutions of the non-autonomous dKdV equation \eqref{eqn:deauto_dkdv_u} and that of Equation \eqref{eqn:intro_multiquadratic_2} with solutions of a non-autonomous lmKdV equation, as shown by the following lemma.
\begin{lemma}
Let $\phi=\phi_{l,m}$ be a solution of the non-autonomous form of lmKdV\cite{NC1995:MR1329559,NQC1983:MR719638,ABS2003:MR1962121}:
\begin{equation}\label{eqn:dkdv_phi_lmkdv_ful_para}
 \dfrac{~\widetilde{\overline{\phi}}~}{\phi}
 =\dfrac{\al_l\,\overline{\phi}-\be_m\,\widetilde{\phi}}{\al_l\,\widetilde{\phi}-\be_m\,\overline{\phi}},
\end{equation}
where
\begin{equation}
 \al_l=\Big(p_l-\la\Big)^{-1/2},\quad
 \be_m=\Big(q_m-\la\Big)^{-1/2}.
\end{equation}
Then, $u$ and $v$ given by
\begin{equation}
 u=\dfrac{\Big(p_l-\la\Big)^{1/2}\overline{\phi}-\Big(q_m-\la\Big)^{1/2}\widetilde{\phi}}{\phi},\quad
 v=\Big(p_l-\la\Big)^{1/2}\dfrac{~\overline{\phi}~}{\phi},
\end{equation}
solve Equations \eqref{eqn:deauto_dkdv_u} and \eqref{eqn:intro_multiquadratic_2}, respectively.
\end{lemma}
\begin{proof}
The statement can be verified by direct calculation.
\end{proof}

\subsection{Properties of Equation \eqref{eqn:redbook390_u}}
The system of {\PDE}s for the lattice sine-Gordon equation \eqref{eqn:redbook390_u}, which has the CABC and tetrahedron properties, is given by:
\begin{subequations}\label{eq:redbooktransf}
\begin{align}
 &{\mathsf A}
 =\dfrac{~\widetilde{\overline{u}}~}{u}-\dfrac{(\ga \widetilde{u}-1)(\ga-\overline{u})}{(\ga-\widetilde{u})(\ga\overline{u}-1)}
 =0,\\
 &{\mathsf S}
 =\Big(1-\ga\overline{u}\Big)\Big(\la+\ga\widetilde{v}\Big)
 -\Big(\ga-\overline{u}\Big)\Big(\ga -\widetilde{\overline{v}}\Big)u\widetilde{v}
 =0,\\
 &{\mathsf B}
 =\Big(\la+\ga \widetilde{v}-u\widetilde{v}\Big)\Big(1-\ga u+uv\Big)+\Big(1-\ga^2-\la\Big)uv
 =0,\\
 &{\mathsf C}
 =\Big(1-\ga u\Big)\Big(\la+\ga v\Big)-\Big(\ga-u\Big)\Big(\ga-\overline{v}\Big)\overline{u}v
 =0.
\end{align}
\end{subequations}
Moreover, the tetrahedron equations are given by 
\begin{subequations}
\begin{align}
 {\mathcal K}_1
 =&\Big(\la-\ga \la u+(1-\ga^2)uv\Big)
 \Big(1-\ga^2+\ga \widetilde{\overline{v}}-\overline{u}\,\widetilde{\overline{v}}\Big)\notag\\
 &+\Big(1-\ga^2-\la\Big)\Big(\ga-u\Big)\Big(1-\ga\overline{u}\Big)v
 =0,\\
 {\mathcal K}_2
 =&\Big(\la+\ga\widetilde{v}-u\widetilde{v}\Big)\Big(1-\ga\overline{u}+\overline{u}\,\overline{v}\Big)
 +\Big(1-\ga^2-\la\Big)u\widetilde{v}
 =0,
\end{align}
\end{subequations}
and the Lax pair of Equation \eqref{eqn:redbook390_u} is given by
\begin{equation}\label{eqn:redbook390_Lax_Psi}
 \overline{\Psi}
 =\begin{pmatrix}
 \ga+\dfrac{\ga(\ga u-1)}{(\ga-u)\overline{u}}&\dfrac{\la(\ga u-1)}{(\ga-u)\overline{u}}\\
 1&0
 \end{pmatrix}
 \Psi,\quad
 \widetilde{\Psi}
 =\begin{pmatrix}
 \dfrac{(\ga^2-1)u}{\ga-u}&\dfrac{\la(\ga u-1)}{\ga-u}\\
 u&1-\ga u
 \end{pmatrix}
 \Psi.
\end{equation}

\begin{theorem}
If $u$ and $v$ satisfy the system of equations \eqref{eq:redbooktransf}, then $v$ satisfies the multi-quadratic equation \eqref{eqn:intro_multiquadratic_3}. 
\end{theorem}
\begin{proof}
The method of proof is the same as in the proof of Theorem \ref{theorem:dkdv_mul}.
\end{proof}

\begin{remark}
The multivariate polynomial ${\mathcal A}'$ for Equation \eqref{eqn:intro_multiquadratic_3}$:$
\begin{align}\label{eqn:redbook390_v_P5}
 {\mathcal A}'\Big(v,\overline{v},\widetilde{v},\widetilde{\overline{v}}\Big)
 =&\ga^2\Big(v-\widetilde{\overline{v}}\Big)\Big(\overline{v}-\widetilde{v}\Big)\Big(\la+v\overline{v}\Big)\Big(\la+\widetilde{v}\,\widetilde{\overline{v}}\Big)
 -\Big(1-\ga^2\Big)\Big(\ga^2-\la\Big)\Big(v\overline{v}-\widetilde{v}\,\widetilde{\overline{v}}\Big)^2\notag\\
 &+\ga\Big(1-\ga^2\Big)\Big(v\overline{v}-\widetilde{v}\,\widetilde{\overline{v}}\Big)
 \Biggm(\Big(\la+v\overline{v}\Big)\Big(\widetilde{v}+\widetilde{\overline{v}}\Big)-\Big(v+\overline{v}\Big)\Big(\la+\widetilde{v}\,\widetilde{\overline{v}}\Big)\Biggm)
 =0
\end{align}
satisfies the Kleinian symmetry \eqref{eqn:Kleinian symmetry},
and its discriminants are given by \eqref{eqns:A'_discriminants} with
\begin{subequations}
\begin{align}
 &G(x,y)=\ga^2(x-y)^2,\quad
 H_1(x,y)=(xy+\la)^2,\\
 &H_2(x,y)=\Big(\ga xy+(1-\ga^2)(x+y)-\ga\la\Big)^2-4\Big(1-\ga^2-\la\Big)xy.
\end{align}
\end{subequations}
From the information given above, we find that Equation \eqref{eqn:redbook390_v_P5} is an equation of $dQ4^{\ast}$-type in the sense of \cite{kassotakis2018difference,atkinson2014multi}.
\end{remark}

The solutions of Equations \eqref{eqn:redbook390_u} and \eqref{eqn:intro_multiquadratic_3} are related to another \PDE \ as shown by the following lemma.
\begin{lemma}
Let $\phi=\phi_{l,m}$ be a solution of the following {\PDE}:
\begin{equation}\label{eq:strangeexample}
 \dfrac{\Big(\ga\widetilde{\overline{\phi}}+\la\phi\Big)\Big(\overline{\phi}-\ga\widetilde{\phi}\Big)}
 {\Big(\widetilde{\overline{\phi}}-\overline{\phi}\Big)\Big(\widetilde{\phi}-\phi\Big)}
 =1-\ga^2.
\end{equation}
Then, $u$ and $v$ given by
\begin{equation}
 u=\dfrac{\phi-\widetilde{\phi}}{\ga\phi-\overline{\phi}},\quad
 v=\dfrac{~\overline{\phi}~}{\phi},
\end{equation}
solve Equations \eqref{eqn:redbook390_u} and \eqref{eqn:intro_multiquadratic_3}, respectively.
\end{lemma}
\begin{proof}
The statement can be verified by direct calculation.
\end{proof}

We note that Equation \eqref{eq:strangeexample} requires further study.

\section{Concluding remarks}\label{ConcludingRemarks}
Much of what we know about the real world is modelled through continuous mathematical equations that become discrete equations on the computer. 
Hirota's dKdV equation is an important example in the study of such discretizations, because it shares the distinctive properties of the KdV equation.

However, there are gaps in its study. 
The main open question studied in this paper concerns its embedding in a three-dimensional lattice and the question of its consistency.
By finding previously unknown transformations to other \PDE s, we show that there is an unusual embedding into a three-dimensional lattice along with a consistency property, which we call {\em consistency around a broken cube}. 
By using this property to construct a Lax pair for the dKdV equation, we show that the embedding is related to its integrability.
It is interesting to note that a previously unknown transformation to a multi-quadratic lattice equation also arises from this construction.

These observations lead to several open questions. 
One is how the construction may extend to higher dimensional lattices, i.e., $\bbZ^N$, where $N\ge 4$. 
A second important question is whether other integrable lattice equations, which do not satisfy the CAC property, turn out to satisfy the CABC property defined in Section \ref{s:3D}. 
We anticipate that there may be other generalizations of consistency that remain to be found, particularly when we consider non-scalar, multi-component \PDE s. 

\subsection*{Acknowledgment}
N. Nakazono would like to thank Dr. P. Kassotakis for inspiring and fruitful discussions about multi-quadratic equations and Dr. Y. Sun for those about the non-autonomous form of the dKdV equation \eqref{eqn:deauto_dkdv_u}.

\def\cprime{$'$} \def\cprime{$'$}

\end{document}